\documentclass[aip,amsmath,amssymb,reprint,jmp,onecolumn]{revtex4-1}\pdfoutput=1
\usepackage{graphicx} 
\usepackage{dcolumn} 
\usepackage{bm} 

\usepackage[utf8]{inputenc}
\usepackage[T1]{fontenc}
\usepackage{mathptmx}
\usepackage[bookmarksnumbered,hypertexnames=false,colorlinks=true,linkcolor=blue,urlcolor=blue,citecolor=blue,anchorcolor=green,breaklinks=true,pdfusetitle]{hyperref}
\usepackage{amsthm,braket,mathtools,enumitem}
\usepackage[USenglish]{babel}
\usepackage[capitalize]{cleveref}

\newtheorem{thm}{Theorem}\crefname{thm}{Theorem}{Theorems}
\newtheorem{lem}[thm]{Lemma}\crefname{lem}{Lemma}{Lemmas}
\crefname{cor}{Corollary}{Corollaries}

\DeclareMathOperator{\tr}{tr}
\DeclareMathOperator{\rk}{rk}
\DeclareMathOperator{\GHZ}{GHZ}
\DeclarePairedDelimiter{\abs}{\lvert}{\rvert}

\newcommand{\EE}{\mathbb E}
\newcommand{\PP}{\mathbb P}
\newcommand{\CC}{\mathbb C}
\newcommand{\RR}{\mathbb R}
\newcommand{\ZZ}{\mathbb Z}
\newcommand{\NN}{\mathbb N}

\newcommand{\cH}{\mathcal H}
\newcommand{\ot}{\otimes}
\newcommand{\eps}{\varepsilon}

\newcommand{\Chyper}[1]{C^H_{#1}}
\newcommand{\Cgraph}[1]{C^G_{#1}}
\newcommand{\Cquantum}[1]{C^Q_{#1}}
\newcommand{\Cstab}[1]{C^S_{#1}}
\usepackage{xcolor}

\begin{document}

\title{Hypergraph min-cuts from quantum entropies}
\author{Michael Walter}
\email{m.walter@uva.nl}
\affiliation{Korteweg-de Vries Institute for Mathematics and QuSoft, University of Amsterdam, the Netherlands}
\affiliation{Institute for Theoretical Physics, Institute for Language, Logic, and Computation, University of Amsterdam, the Netherlands}
\author{Freek Witteveen}
\email{f.g.witteveen@uva.nl}
\affiliation{Korteweg-de Vries Institute for Mathematics and QuSoft, University of Amsterdam, the Netherlands}
\date{\today}

\begin{abstract}
The von Neumann entropy of pure quantum states and the min-cut function of weighted hypergraphs are both symmetric submodular functions.
In this article, we explain this coincidence by proving that the min-cut function of any weighted hypergraph can be approximated (up to an overall rescaling) by the entropies of quantum states known as stabilizer states.
We do so by constructing a novel ensemble of random quantum states, built from tensor networks, whose entanglement structure is determined by a given hypergraph.
This implies that the min-cuts of hypergraphs are constrained by quantum entropy inequalities, and it follows that the recently defined hypergraph cones are contained in the quantum stabilizer entropy cones, which confirms a conjecture made in the recent literature.
\end{abstract}

\maketitle

\section{Quantum states and hypergraphs}
Given a quantum state~$\rho$ on a finite-dimensional tensor product Hilbert space $\cH = \bigotimes_{t\in T} \cH_t$, where $T:=[n]:=\{1,\dots,n\}$ is the set of subsystem labels, its \emph{entropy function} $S\colon 2^T\to\RR_{\geq0}$ assigns to each subset $A\subseteq T$ the von Neumann entropy $S(A) = -\tr[\rho_A \log \rho_A]$ of the reduced state~$\rho_A = \tr_{T\setminus A}[\rho]$.
A basic property of the entropy function is that if $\rho$ is pure, the entropy function is \emph{symmetric}, i.e., $S(A) = S(T \setminus A)$ for all $A \subseteq T$.
By a celebrated theorem of Lieb-Ruskai, the entropy function is also \emph{submodular}, meaning that
\begin{align*}
  S(A) + S(B) \geq S(A \cup B) + S(A \cap B) \qquad \forall A, B \subseteq T.
\end{align*}
This inequality is also known as the strong subadditivity of the von Neumann entropy.
It is interesting to ask if there are any other entropy inequalities satisfied by arbitrary quantum states.
To formalize this, one can define the \emph{quantum entropy cone}~$\Cquantum n \subseteq \RR_{\geq0}^{2^n}$ as the closure of the set of entropy functions obtained by varying over all pure quantum states on finite-dimensional tensor product Hilbert spaces as above.
The result is a closed convex cone.\cite{pippenger2003inequalities}
It is a well-known open problem in quantum information theory to determine the cones~$\Cquantum n$ for~$n \geq 5$.\cite{pippenger2003inequalities,linden2005new}

One strategy to make progress on this difficult problem has been to restrict
to certain classes of quantum states, in order to obtain inequalities that only hold for these specific classes of quantum states.
One example is the class of stabilizer states, which are a versatile family of quantum states that have many applications in quantum information theory.\cite{gottesman1997stabilizer,gross2017schur}
Thus, let $\Cstab n$ denote the closed convex cone generated by the set of entropy functions of pure stabilizer states (over any fixed prime).
In general, $\Cstab n \subseteq \Cquantum n$ is a proper subcone.\cite{linden2013quantum,gross2013stabilizer}
In other words, there are entropy inequalities which are valid for all stabilizer states, but may be violated by general quantum states.

Another class of states of interest are states arising from holographic space-times in the sense of the AdS/CFT correspondence.\cite{maldacena1999large}
In that case it is known that (up to small corrections) the entropy function has a geometric interpretation as the size of certain surfaces, as prescribed by the Ryu-Takayanagi formula.\cite{ryu2006holographic}
This has served as motivation to define the holographic entropy cone, where one considers the cone of entropies as computed by the Ryu-Takayanagi formula on some arbitrary space.\cite{bao2015holographic}
This cone has been used to prove various new entropy inequalities for holographic states\cite{bao2015holographic} and is a subject of intense current research in theoretical high energy physics.\cite{bao2015gapped,marolf2017handlebody,cuenca2019holographic,czech2019holographic,hubeny2019holographic}

The holographic entropy cone can be shown to be equal to the cone generated by min-cut functions on graphs.\cite{bao2015holographic}
We will now introduce a generalization, by considering \emph{hypergraphs}.
Given a hypergraph $G=(V,E)$ and non-negative hyperedge weights~$w\colon E\to\RR_{\geq0}$, the \emph{cut function} $c \colon 2^V\to\RR_{\geq0}$ is defined as $c(M) = \sum_{e\in\delta(M)} w(e)$, where $\delta(M)$ is the set of hyperedges that contain vertices both in~$M$ and~$V\setminus M$.
Fixing a subset of terminals~$T\subseteq V$, the \emph{min-cut function} $m\colon 2^T \to \RR_{\geq0}$ is then given by $m(A) = \min_{M : M \cap T = A} c(M)$.
If all hyperedges have cardinality two then the above reduce to the usual notions for weighted undirected graphs.
It is well-known that hypergraph cut functions are symmetric and submodular, just like the quantum entropy function.
This property extends directly to min-cut functions, so it holds that
\begin{align*}
  m(A) &= m(T \setminus A) \qquad \forall A \subseteq T, \\
  \quad m(A) + m(B) &\geq m(A \cup B) + m(A \cap B) \qquad \forall A, B \subseteq T.
\end{align*}
Similar as before we may define the \emph{hypergraph cone}~$\Chyper n \subseteq \RR_{\geq0}^{2^n}$ as the set of min-cut functions obtained from arbitrary weighted hypergraphs with terminals~$T=[n]$, which again is a closed convex cone.\cite{bao2020quantum}
The holographic entropy cone mentioned above can be defined in similar fashion, restricting to min-cut functions of graphs rather than general hypergraphs; we denote this cone by~$\Cgraph n \subseteq \RR_{\geq0}^{2^n}$.

The coincidence that both the min-cut function of hypergraphs and the entropy function of quantum states are symmetric and submodular begs the question if hypergraph min-cuts can always be realized by quantum entropies (or vice versa).
In other words, how are $\Cquantum n$ and $\Chyper n$ related?
Further motivation for studying this problem in the context of holography was given in Ref.~\onlinecite{bao2020quantum}.

Our main result is as follows:

\begin{thm}\label{thm:main}
  For any $n\in\NN$, we have that $\Chyper n \subseteq \Cstab n \subseteq \Cquantum n$.
\end{thm}

\noindent
\cref{thm:main} was conjectured to hold in Ref.~\onlinecite{bao2020quantum}, where the authors explicitly verified the inclusion for~$n \leq 5$.
We note that for larger $n$ not even the inclusion $\Chyper n \subseteq \Cquantum n$ was known before our result.
Shortly after our result, follow-up work showed that $\Chyper 6 \neq \Cstab 6$.\cite{bao2020gap}
Together with \cref{thm:main} this implies that the inclusion~$\Chyper n \subseteq \Cstab n$ is strict for $n\geq 6$.

To prove \cref{thm:main}, we need to construct, for any weighted hypergraph, a quantum state whose entropies realize the min-cut function.
To achieve this, we use the hypergraph to define a tensor network structure, with multipartite entangled states on the hyperedges rather than maximally entangled states.
The min-cut function gives an upper bound on the entropy function of the state obtained by placing arbitrary tensors on the vertices in this network.
By choosing random (stabilizer) tensors we obtain an ensemble of (stabilizer) quantum states which satisfies a generalized version of the Ryu-Takayanagi formula; namely the upper bound is almost satisfied and hence the entropy function is approximated by the min-cut function on the hypergraph (with high probability).
By increasing the local dimension in the tensor network we can approximate the hypergraph min-cut function to arbitrary precision, which then proves the inclusion $\Chyper n \subseteq \Cstab n$.
In Refs.~\onlinecite{hayden2016holographic,nezami2016multipartite}, it was shown that $\Cgraph n \subseteq \Cstab n \subseteq \Cquantum n$ by considering random tensor network states.
Since by definition $\Cgraph n \subseteq \Chyper n$, \cref{thm:main} strengthens this result.
Our method of proof suggests that various other results of Ref.~\onlinecite{hayden2016holographic} can be extended to the setting of random tensor networks on hypergraphs, such as the possibility to include bulk states and to perform error correction on the entanglement wedge.
Another interesting extension would be to study the effect of the multipartite entanglement included in our hypergraph tensor network ansatz on the multipartite entanglement properties of the resulting quantum state, as was studied for graphs in Ref.~\onlinecite{nezami2016multipartite}.

\section{Random tensor network states for hypergraphs min-cuts}
Fix a hypergraph $G=(V,E)$ with integral hyperedge weights~$w\colon E \to \ZZ_{\geq0}$.
For each vertex $x\in V$, define the Hilbert space $\cH_x = \bigotimes_{e \in E : x \in e} \cH_{x,e}^{\ot w(e)}$, where $\cH_{x,e} = \CC^D$.
The dimension $D$ will later be taken to be large.
Now let $\ket\Omega$ be the state given by
\begin{align}\label{eq:omega}
    \ket\Omega = \bigotimes_{e\in E} \ket{\GHZ(e)}^{\ot w(e)} \in \bigotimes_{x\in V} \cH_x.
\end{align}
Here, $\ket{\GHZ(e)} = \frac1{\sqrt D} \sum_{i=1}^D \ket i^{\ot \abs e} \in \bigotimes_{x \in e} \cH_{x,e}$ denotes an $\abs e$-partite GHZ state of local dimension~$D$.
One can also develop the following theory for other multipartite entangled states than the GHZ states, but for our purposes this choice will suffice.

The GHZ state has the property that it is a pure state whose reduced states have~$D$ nonzero eigenvalues which are all equal to~$1/D$.
This implies that not only the von Neumann entropy of any subsystem $M \subseteq V$, but in fact any R\'enyi-$\alpha$ entropy (which we denote by~$S_\alpha$) can be calculated by
\begin{equation}\label{eq:cut entropies}
\begin{aligned}
    S_\alpha(\Omega_M) &= \sum_{e\in E} w(e) \, S_\alpha\bigl(\GHZ(e)_{M \cap e}\bigr)
= \log(D) \sum_{e\in \delta(M)} w(e)
= \log(D) \, c(M).
\end{aligned}
\end{equation}
We have thus found a family of quantum states, parameterized by $D$, whose entropy function is exactly proportional to the cut function~$c\colon 2^V\to\RR_{\geq 0}$ of the given hypergraph.

We now explain how to construct a quantum state, based on the one in~\eqref{eq:omega}, whose entropy function is related to the \emph{min}-cut function~$m\colon 2^T\to\RR_{\geq 0}$ of the hypergraph $G=(V,E)$ for any fixed choice of terminals~$T\subseteq V$.
We may assume that any connected component of $G$ touches $T$ (otherwise we can remove this component without impacting the min-cut function).
Our main tool is the following construction, which is implicit in Ref.~\onlinecite{hayden2016holographic} (cf.~Refs.~\onlinecite{horodecki2007quantum,dutil2010one}).
Define the (not necessarily normalized) pure state
\begin{align}\label{eq:projected state}
    \ket{\Psi} = \Bigl( \bigotimes_{x \in V \setminus T} \bra{\phi_x} \Bigr) \ket\Omega \in \cH = \bigotimes_{x\in T} \cH_x,
\end{align}
obtained by projecting the tensor factors for each non-terminal vertex~$x\in V\setminus T$ onto pure states~$\ket{\phi_x} \in \mathcal H_x$ (for now these states are arbitrary; we later choose them at random).
We note that~$\Psi$ can be understood as a tensor network state of bond dimension~$D$.

We can relate the entropies of~$\Psi$ to those of the state~$\Omega$.
For this, recall that for any quantum state~$\rho$,
\begin{align}\label{eq:renyi mono}
    S_2(\rho) \leq S(\rho) \leq S_0(\rho),
\end{align}
where $S_2(\rho) = -\log\tr[\rho^2]$ is the R\'enyi-2 entropy, $S(\rho) = -\tr[\rho\log\rho]$ the von Neumann entropy, and $S_0(\rho) = \log\rk[\rho]$ the log-rank.
If $\rho$ is not normalized then we define $S_\alpha(\rho)$ in terms of the normalization (this makes no difference for the log-rank).
Then it is not hard to see that the state~\eqref{eq:projected state} satisfies
\begin{align}\label{eq:rank upper}
    S_0(\Psi_A) \leq \min_{M : M \cap T = A} S_0(\Omega_M) = \log(D) m(A)
    \qquad
    \forall A \subseteq T.
\end{align}
Indeed, for any cut $M$ with $M \cap T = A$ we can upper bound the rank of $\Psi_A$ in terms of the rank of $\Omega_M$, which in turn can be evaluated using~\eqref{eq:cut entropies} for $\alpha=0$.

\section{Proof of Theorem~\ref{thm:main} using random tensor network states}
To prove our main result it suffices, in view of~\eqref{eq:renyi mono}, to complement the upper bound in~\eqref{eq:rank upper} by a similar lower bound on the R\'enyi-2 entropy.
For this, we choose each $\phi_x=\ket{\phi_x}\bra{\phi_x}$ independently at random from a \emph{projective 2-design}.
A 2-design is an ensemble of pure states such that the first two moments agree with the unitarily invariant (`Haar') probability measure on pure states, i.e.,
\begin{align}\label{eq:haar moments}
  \EE[\phi_x] = \frac I {D_x}
  \quad\text{and}\quad
  \EE[\phi_x^{\ot 2}] = \frac{I + F}{D_x(D_x+1)},
\end{align}
where~$D_x = \dim \cH_x$ and~$F$ denotes the swap operator on $\cH_x^{\ot 2}$.
In particular, we may choose each~$\phi_x$ uniformly at random from the finite set of stabilizer states,\cite{gottesman1997stabilizer} which is a well-known 2-design.\cite{klappenecker2005mutually,gross2007evenly}
In this case,~$\Psi$~is again a stabilizer state.\cite{hayden2016holographic}
We now compute the expected trace and purity of any subsystem.

\begin{lem}\label{lem:moments}
Let the $\phi_x$ be chosen independently at random from a 2-design.
Then the state~$\Psi$ in~\eqref{eq:projected state} satisfies
\begin{align*}
    \EE[\tr[\Psi]] = \frac1{D_b}
\end{align*}
and
\begin{align*}
    \EE[\tr[\Psi_A^2]] = \frac1{D_b^2} D^{-m(A)} \left( k_A + O(D^{-1}) \right)
\end{align*}
for all $A \subseteq T$, where $D_b = \prod_{x\in V\setminus T} D_x$ and $k_A$ denotes the number of minimal cuts for~$A$.
\end{lem}
\begin{proof}
From the formula for the first moment of a 2-design in~\eqref{eq:haar moments},
\begin{align*}
  \EE[\tr[\Psi]] = \tr\Bigl[\Omega \bigl( \bigotimes_{x \in V \setminus T} \EE[\phi_x] \bigr) \Bigr] = \frac 1 {D_b}
\end{align*}
(as is common we leave out tensor products with identity operators).
This establishes the first claim.

For the second claim, we first use the swap trick to compute
\begin{align*}
  \tr[\Psi_A^2]
&= \tr[\Psi^{\ot 2} F_A] \\
&= \tr\Bigl[\bigl( \bigotimes_{x \in V \setminus T} \bra{\phi_x}^{\ot 2} \bigr) \Omega^{\ot 2} \bigl( \bigotimes_{x \in V \setminus T} \ket{\phi_x}^{\ot 2} \bigr) F_A\Bigr]\\
&= \tr\Bigl[\Omega^{\ot 2} \bigl( \bigotimes_{x \in V \setminus T} \phi_x^{\ot 2} \bigr) F_A\Bigr].
\end{align*}
Here, $F_A$ denotes the product of the swap operators on~$\cH_x^{\ot2}$ for $x\in A$.
In the last step, we used the cyclicity of the trace and the fact that the swap operator $F_A$ `commutes' with $\bra{\phi_x}^{\ot 2}$ for~$x\in V\setminus T$.
We can then compute the average by using the formula for the second moment of a 2-design in~\eqref{eq:haar moments}:
\begin{align*}
  \EE[\tr[\Psi_A^2]]
&= \tr\Bigl[\Omega^{\ot 2} \bigl( \bigotimes_{x \in V \setminus T} \EE[\phi_x^{\ot 2}] \bigr) F_A\Bigr] \\
&= \tr\Bigl[\Omega^{\ot 2} \bigl( \prod_{x \in V \setminus T} \frac {I_x + F_x} {D_x(D_x+1)} \bigr) F_A\Bigr] \\
&= \prod_{x\in V\setminus T} \frac1{D_x(D_x+1)} \sum_{M : M \cap T = A} \tr\bigl[\Omega^{\ot 2} F_M \bigr]
\end{align*}
where the last step is again by the swap trick.
The prefactor is $D_b^{-2} (1 + O(D^{-1}))$, while the sum can be estimated using~\eqref{eq:cut entropies},
\begin{align*}
  \sum_{M : M \cap T = A} \tr\bigl[\Omega^{\ot 2} F_M \bigr]
&= \sum_{M : M \cap T = A} 2^{-S_2(\Omega_M)}
= \sum_{M : M \cap T = A} D^{-c(M)}
= D^{-m(A)} \left( k_A + O(D^{-1}) \right),
\end{align*}
where we recall that~$k_A$ denotes the number of minimal cuts for $A$.
Together we obtain the second claim.
\end{proof}

For $A=\emptyset$, the min-cut is empty and nondegenerate by our assumption that any connected component of~$G$ touches~$T$, so \cref{lem:moments} implies in particular that $\EE[\tr[\Psi]^2] = D_b^{-2} (1 + O(D^{-1}))$.
Thus, $\tr[\Psi]$ is concentrated around its mean, suggesting that, with high probability, $\Psi\neq0$ and $S_2(\Psi_A) / \! \log(D) \approx m(A)$ for large~$D$.
The following lemma follows the proof strategy of Ref.~\onlinecite{hayden2016holographic} to make this intuition precise.

\begin{lem}\label{lem:concentration}
Let~$\Psi$ be defined as in~\eqref{eq:projected state}, with each $\phi_x$ chosen independently at random from a 2-design. Then the following properties hold for large~$D$:
\begin{enumerate}[label={(\alph*)}]
\item\label{it:nonzero}
$\PP(\Psi\neq0) = 1 - O(D^{-1})$.
\item\label{it:expected renyi}
$\EE[S_2(\Psi_A) | \Psi\neq0] \geq \log(D) \, m(A) - \log(k_A) - O(D^{-1/4})$, with $k_A$ the number of minimal cuts for~$A$.
\item\label{it:high probability entropy}
For any~$\delta>0$ it holds that
\begin{align*}
  \PP\Bigl(\Psi\neq0 \text{ and } \abs[\Big]{ \frac{S_2(\Psi_A)}{\log(D)} - m(A) } \leq \delta \text{ for all } A \subseteq T\Bigr) = 1 - O\Bigl(\frac1{\delta\log(D)}\Bigr).
\end{align*}
The same statement holds for the von Neumann entropy~$S(\Psi_A)$ instead of the R\'enyi entropy~$S_2(\Psi_A)$.
\end{enumerate}
\end{lem}
\begin{proof}
As just noted, $\EE[\tr[\Psi]^2] = \EE[\tr[\Psi]]^2 (1 + O(D^{-1}))$, so Chebyshev's inequality shows that, for any~$\eps>0$,
\begin{align}\label{eq:cheby}
  \PP\bigl(\lvert D_b\tr[\Psi] - 1 \rvert \leq \eps \bigr)
\geq 1 - O\Bigl(\frac1{\eps^2 D}\Bigr).
\end{align}
This establishes~\ref{it:nonzero}.

Next we prove~\ref{it:expected renyi}.
Let $E$ denote the event that $\lvert D_b\tr[\Psi] - 1 \rvert \leq D^{-1/4}$.
By~\eqref{eq:cheby},
\begin{align}\label{eq:p_E}
    p_E := \PP(E) = 1 - O(D^{-1/2}).
\end{align}
Since~$E$ implies that~$\Psi\neq0$, we have
\begin{equation}\label{eq:E[S_2|nonzero]}
\begin{aligned}
  \EE[S_2(\Psi_A) | \Psi\neq0]
&\geq \PP(\Psi\neq0) \, \EE[S_2(\Psi_A) | \Psi\neq0] \\
&\geq p_E \, \EE[S_2(\Psi_A) | E]
\end{aligned}
\end{equation}
We now bound
\begin{equation}\label{eq:E[S_2|E]}
\begin{aligned}
  \EE[S_2(\Psi_A) | E]
&= -\EE\bigl[\log(D_b^2 \tr[\Psi_A^2]) \big| E \bigr] + 2 \EE\bigl[ \log(D_b\tr[\Psi]) \big| E \bigr] \\
&\geq -\log\EE\bigl[D_b^2 \tr[\Psi_A^2] \big| E \bigr] + 2 \log\bigl(1-D^{-1/4}\bigr) \\
&= -\log\EE\bigl[D_b^2 \tr[\Psi_A^2] \big| E \bigr] - O(D^{-1/4})
\end{aligned}
\end{equation}
where we used Jensen's inequality to lower-bound the first term. 
Using~$p_E \, \EE[\tr[\Psi_A^2] | E] \leq \EE[\tr[\Psi_A^2]]$, we obtain
\begin{align*}
  -\log\EE\bigl[D_b^2 \tr[\Psi_A^2] \big| E \bigr]
&\geq -\log \EE[D_b^2 \tr[\Psi_A^2]] + \log(p_E) \\
&= \log(D) \, m(A) - \log \bigl( k_A + O(D^{-1}) \bigr) + \log\bigl( 1 - O(D^{-1/2}) \bigr) \\
&= \log(D) \, m(A) - \log(k_A) - O(D^{-1/2})
\end{align*}
by \cref{lem:moments} and~\eqref{eq:p_E}.
Together with~\eqref{eq:E[S_2|nonzero]}, \eqref{eq:E[S_2|E]}, and~$S_2(\Psi_A) \leq \log(D)m(A)$, which holds by~\eqref{eq:renyi mono} and \eqref{eq:rank upper}, we find
\begin{align*}
  \EE[S_2(\Psi_A) | \Psi\neq0]
&\geq \log(D) \, m(A) - \log(k_A) - O(D^{-1/4}),
\end{align*}
proving~\ref{it:expected renyi}.

To prove \ref{it:high probability entropy}, we note that $\log(D)m(A) - S_2(\Psi_A)$ is a nonnegative random variable. Thus, for any fixed~$A\subset T$
\begin{align*}
\PP\Bigl(m(A) - \frac{S_2(\Psi_A)}{\log(D)} > \delta \Big| \Psi\neq0 \Bigr)
&\leq \frac{\log(D) m(A) - \EE[S_2(\Psi_A) | \Psi\neq0 ]}{\delta \log(D)} \\
&\leq \frac{\log(k_A) + O(D^{-1/4})}{\delta \log(D)}\\
&= O\left(\frac1{\delta\log(D)}\right)
\end{align*}
where we first used the Markov inequality and then part~\ref{it:expected renyi}.
By taking the union bound over all the finitely many subsets~$A \subseteq T$ and using part~\ref{it:nonzero}, we obtain that
\begin{align*}
  \PP\Bigl(\Psi\neq0 \text{ and } m(A) - \frac{S_2(\Psi_A)}{\log(D)} \leq \delta \text{ for all } A \subseteq T\Bigr) = 1 - O\Bigl(\frac1{\delta\log(D)}\Bigr).
\end{align*}
In view of $S_2(\Psi_A) \leq S(\Psi_A) \leq \log(D)m(A)$, this proves part~\ref{it:high probability entropy}.
\end{proof}

We finally prove our main result, which now follows readily from \Cref{lem:concentration}.

\begin{proof}[Proof of \cref{thm:main}]
We only need to show that $\Chyper n \subseteq \Cstab n$.
Since $\Cstab n$ is a closed cone, it suffices to show that for any $\delta>0$ and any hypergraph with integral hyperedge weights, terminal set~$T=[n]$, and min-cut function~$m$, there exists a number~$c>0$ and a stabilizer state $\Psi$ on an $n$-partite Hilbert space such that
\begin{align*}
  \abs[\Big]{ \frac{S(\Psi_A)}{c} - m(A) } \leq \delta
\end{align*}
for all~$A \subseteq T$.
This follows from \cref{lem:concentration} if we use the set of stabilizer states as the 2-design and choose~$D$ to be sufficiently large.
Indeed, if each~$\phi_x$ is a stabilizer state then so is $\Psi$, as discussed above.
\end{proof}

\subsection*{Acknowledgements}
We would like to thank Ning Bao and Newton Cheng for interesting correspondence on related topics.
MW is supported by NWO grants Veni~680-47-459 and OCENW.KLEIN.267.

\subsection*{Data availability}
Data sharing is not applicable to this article as no new data were created or analyzed in this study.

\bibliography{hypergraphs}

\end{document}